\documentclass[11pt,letterpaper]{article}
\usepackage[pass]{geometry}

\usepackage{palatino}
\usepackage{mathpazo}
\usepackage{stmaryrd}

\usepackage{hyperref}
\hypersetup{pdfpagemode=UseNone}

\usepackage{amsfonts}
\usepackage{amssymb}
\usepackage{amsmath}
\usepackage{latexsym}
\usepackage{amsthm}
\usepackage{sectsty}

\newtheorem{theorem}{Theorem}
\newtheorem{lemma}[theorem]{Lemma}
\newtheorem{prop}[theorem]{Proposition}

\theoremstyle{definition}

\newcommand{\tinyspace}{\mspace{1mu}}
\newcommand{\microspace}{\mspace{0.5mu}}
\newcommand{\op}[1]{\operatorname{#1}}

\newcommand{\norm}[1]{\left\lVert\tinyspace#1\tinyspace\right\rVert}

\newcommand{\tr}{\operatorname{Tr}}

\newcommand{\ip}[2]{\left\langle #1 , #2\right\rangle}

\def\({\left(}
\def\){\right)}
\def\I{\mathbb{1}}

\newcommand{\setft}[1]{\mathrm{#1}}
\newcommand{\lin}[1]{\setft{L}\left(#1\right)}

\newcommand{\density}[1]{\setft{D}\left(#1\right)}

\newcommand{\herm}[1]{\setft{Herm}\left(#1\right)}
\newcommand{\pos}[1]{\setft{Pos}\left(#1\right)}

\newcommand{\pd}[1]{\setft{Pd}\left(#1\right)}

\def\complex{\mathbb{C}}
\def\real{\mathbb{R}}

\def\integer{\mathbb{Z}}

\def \lket {\left|}
\def \rket {\right\rangle}
\def \lbra {\left\langle}
\def \rbra {\right|}
\newcommand{\ket}[1]{\lket\microspace #1 \microspace\rket}
\newcommand{\bra}[1]{\lbra\microspace #1 \microspace\rbra}

\newenvironment{mylist}[1]{\begin{list}{}{
	\setlength{\leftmargin}{#1}
	\setlength{\rightmargin}{0mm}
	\setlength{\labelsep}{2mm}
	\setlength{\labelwidth}{8mm}
	\setlength{\itemsep}{0mm}}}
	{\end{list}}

\newcommand{\reg}[1]{\mathsf{#1}}

\def\X{\mathcal{X}}
\def\Y{\mathcal{Y}}
\def\Z{\mathcal{Z}}
\def\W{\mathcal{W}}

\def\E{\mathcal{E}}

\newcommand{\footremember}[2]{%
   \footnote{#2}
    \newcounter{#1}
    \setcounter{#1}{\value{footnote}}%
}
\newcommand{\footrecall}[1]{%
    \footnotemark[\value{#1}]%
} 

\begin{document}
\title{Optimal counterfeiting attacks and generalizations for
  Wiesner's quantum money}

\author{
  Abel Molina,\footremember{1}{%
    Institute for Quantum Computing and School of Computer Science, 
    University of Waterloo.
    Supported by Canada's NSERC, CIFAR, and MITACS.} 
  \quad
  Thomas Vidick,\footnote{%
    Computer Science and Artificial Intelligence Laboratory, 
    Massachusetts Institute of Technology. 
    Supported by the National Science Foundation under Grant No. 0844626.}
  \quad and\quad
  John Watrous\footrecall{1}
}

\maketitle       

\begin{abstract}
  We present an analysis of Wiesner's quantum money scheme, as well as
  some natural generalizations of it, based on semidefinite
  programming.
  For Wiesner's original scheme, it is determined that the optimal
  probability for a counterfeiter to create two copies of a bank note
  from one, where both copies pass the bank's test for validity, is
  $(3/4)^n$ for $n$ being the number of qubits used for each note.
  Generalizations in which other ensembles of states are substituted
  for the one considered by Wiesner are also discussed, including
  a scheme recently proposed by Pastawski, Yao, Jiang, Lukin, and
  Cirac, as well as schemes based on higher dimensional quantum
  systems.
  In addition, we introduce a variant of Wiesner's quantum money in
  which the verification protocol for bank notes involves only
  classical communication with the bank. We show that the optimal
  probability with which a counterfeiter can succeed in two
  independent verification attempts, given access to a single valid
  $n$-qubit bank note, is $(3/4+\sqrt{2}/8)^n$.
  We also analyze extensions of this variant to higher-dimensional
  schemes.
\end{abstract}

\section{Introduction}

Wiesner's protocol for quantum money \cite{Wiesner83} 
was a formative idea
in quantum information processing.
In this protocol, a bank generates a bank note composed of $n$
qubits: each qubit is initialized to a state chosen uniformly at
random from the set $\{ \ket{0}, \ket{1}, \ket{+}, \ket{-}\}$, and
this choice of states is kept secret by the bank.
The bank can later check the authenticity of a given note by
performing a measurement on each of its qubits, in accordance with its
secret record of their original states.
(Each bank note is labeled with a unique serial number, so that all of
the bank notes in circulation may be treated independently.)
The security of Wiesner's scheme rests on the principle that quantum
states cannot be cloned---that is, a malicious attacker, given access
to a fixed supply of authentic bank notes, cannot generate a
\emph{larger} quantity of valid bank notes than those to which he was
initially given access.

Although Wiesner's scheme was introduced almost three decades ago, to
the best of our knowledge no rigorous analysis with explicit
bounds on the security of the scheme exists in the literature.
The intuition that the scheme's security follows from the no-cloning
principle appears in \cite{LoSP98}, and quantitatively one \emph{should}
be able to obtain exponential security guarantees from results
such as proofs of the security of the BB84 quantum key exchange
protocol~\cite{BennettB84,ShorP00,May01} or of
uncloneable encryption~\cite{Gott02}.
In this paper we prove tight bounds on the security of Wiesner's
quantum money scheme, through a simple and easily extended argument
based on semidefinite programming.

We consider the specific situation in which a counterfeiter, given
access to a \emph{single} authentic bank note, attempts to create
\emph{two} bank notes having the same serial number that independently
pass the bank's test for validity.
We will call such attacks \emph{simple counterfeiting attacks}. 
Our first main result is the following. 

\begin{theorem}\label{thm:main-quant}
  The optimal simple counterfeiting attack against
  Wiesner's quantum money scheme has success probability exactly
  $(3/4)^n$, where $n$ is the number of qubits in each bank
  note.\footnote{%
    Wiesner~\cite{Wiesner83} in fact arrived at a similar bound, but
    through a not-so-rigorous argument!}
\end{theorem}

\noindent
Other types of attacks are not analyzed in this paper, but we must
note their existence!
For instance, a counterfeiter may attempt to create or copy bank notes
through multiple interactions with the bank.
One simple example of such an attack does not require counterfeiters
to possess any bank notes to start with: by substituting one of two
qubits of a Bell state for each qubit of a bank note, a counterfeiter
can succeed in passing the bank's test for validity with probability
$2^{-n}$, and then conditioned on having succeeded the counterfeiter
will be guaranteed to hold a second valid bank note.\footnote{%
  Lutomirski \cite{Lutomirski10} considered a related scenario where
  the bank kindly provides counterfeiters with access to a bank note's
  post-measurement qubits, regardless of whether validity was
  established. He proved that $O(n)$ verification attempts are
  sufficient to break the protocol in this setting.}
One would therefore expect that the bank would charge a small fee
for testing validity, or 
perhaps alert the authorities when an individual repeatedly makes
failed attempts to validate bank notes,
 for otherwise counterfeiters have a very small but
positive incentive to attack the protocol.
Generally speaking, an analysis of attacks of this nature would seem
to require a limit on the number of verification attempts permitted,
or the specification of a utility function that weighs the potential
gain from counterfeiting against the costs for multiple verifications.
We expect that the semidefinite programming method used to prove
Theorem~\ref{thm:main-quant} would be useful for analyzing such
attacks, but we leave this as a problem for interested readers to
consider.

We also consider simple counterfeiting strategies against quantum
money schemes that generalize Wiesner's original scheme.
These are the schemes obtained by varying the set of possible states that
a quantum bank note may store, as well as the underlying probabilities
for those states. We show that there is a scheme based on the repetition
of a $4$-state single-qubit scheme (i.e., having the same structure as
Wiesner's) for which the optimal simple counterfeiting attack has
success probability $(2/3)^n$, which is optimal among all schemes of
that form. 
Furthermore, we show that any money scheme based on the
use of $d$-dimensional bank notes is subject to a simple
counterfeiting attack with success probability at least $2/(d+1)$, and
we describe a scheme for which this is the best one can do. 

\medskip

One drawback of Wiesner's money scheme is that, not only does it
involve communicating with a centralized bank in order to establish
the authenticity of a given bank note,\footnote{%
  There has also been work in recent years on creating quantum money
  schemes that do not require any communication with the bank in order
  to verify a bank note, but this is only possible under computational
  assumptions~\cite{FarhiGHLS10,LAFGHKS09,Aar09}.}
but it also requires \emph{quantum} communication: bills have to be
sent to the bank for verification.
Gavinsky~\cite{Gav11} recently introduced an alternative scheme in
which bills can be authenticated using only \emph{classical}
communication with the bank.

We consider the following procedure for classical verification of an
$n$-qubit bank note, constructed as in Wiesner's scheme. 
The bank sends the user a random challenge $c\in\{0,1\}^n$.
An honest user should measure the $i$-th qubit in the computational
basis if $c_i=0$, or in the Hadamard basis if $c_i=1$, and send the
measurement outcomes $b\in\{0,1\}^n$ to the bank. 
The bank validates the bank note if and only if whenever $c_i$
corresponded to the basis in which qubit $i$ was encoded, $b_i$
describes the correct outcome.
(A similar scheme was independently introduced recently
in~\cite{PYJLC11}.) 
In this setting, a \emph{simple counterfeiting attack} is one in which a
counterfeiter tries to succeed in \emph{two} independent
authentications with the bank, given access to a single valid bank
note.
Our second main result is the following.

\begin{theorem}\label{thm:main-class}
  For the classical-verification analogue of Wiesner's quantum money
  scheme, the optimal simple counterfeiting attack has success
  probability exactly  $\big(3/4+\sqrt{2}/8\big)^n$, for $n$ being the
  number of qubits in each bank note. 
\end{theorem} 

\noindent
As for Theorem~\ref{thm:main-quant}, our proof of Theorem~\ref{thm:main-class}
follows from the use of semidefinite programming techniques.
In addition we show that, contrary to the quantum-verification
setting, Wiesner's scheme is optimal as long as one considers only
qubits: either changing the bases used to encode each qubit or
increasing the number of possible bases will not improve the scheme's
security against simple counterfeiting attacks.
We also consider a natural generalization of this scheme to bank notes
made of $d$-dimensional qudits, and prove that the optimal simple
counterfeiting attack against it has success probability
exactly $(3/4+1/(4\sqrt{d}))^n$. 

\paragraph{Related work.}
The no-cloning theorem~\cite{WoottersZ82} states that there is no perfect
quantum cloning machine. 
This impossibility result relies on two assumptions: that we are
trying to clone \emph{all} possible states (of a given dimension), and
that we are trying to do so \emph{perfectly}.
Relaxing either or both assumptions opens the way for a fruitful
exploration of the possibility of \emph{approximate} cloning
machines. 
Most work in this area focuses on obtaining \emph{universal}
cloners---required to work for all possible input states---but that
may not be perfect. 

To quantify the quality of a cloner one has to settle on a
\emph{figure of merit}.
Two main figures have been considered: the minimum 
(or, alternately, the average) overlap between one of the two output
clones with the input state, or the joint overlap of both output clones
with a tensor product of the input state with itself.\footnote{%
  In both cases, the specific distance measure used can also be
  varied.
  For instance, the trace distance and the Hilbert-Schmidt distance on
  density matrices have been considered.} 
Bu\v{z}ek and Hillery~\cite{BuzekH96} determined the optimal universal
qubit cloner in the first case, while Werner~\cite{Werner98} solved
the general problem with respect to the second figure of merit.

In the setting of quantum money, however, the first assumption is also
relaxed: a counterfeiter only needs to be successful in cloning the
specific states that are used to create the bank notes.
Work in this direction includes that of Bru\ss\ et
al.~\cite{BruCDM00}, who determined the optimal cloner for
the states used in Wiesner's original money scheme, and for the first
figure of merit discussed above. 
While in this work we consider the second figure of merit, which is
the one appropriate to the context of quantum money, our results can
easily be extended to the first. 

We use a semidefinite programming formulation of the problem, in which one can numerically determine the success probability of an optimal cloner, given any desired possible set of input states and underlying distribution. The connection between cloning of quantum states and semidefinite
programming was observed by Audenaert and De~Moor~\cite{AudenaertD02},
and has been used in the study of cloning by other researchers.
(See, for instance, the survey of Cerf and 
Fiur\'{a}\v{s}ek \cite{CerfF06}.) 
The formulation that we use is closely related to one used in
\cite{MolinaW11}, and can also be seen as a special case of a
semidefinite programming framework for more general quantum strategies
developed in~\cite{GutoskiW07}.

Recent work of Pastawski et al.~\cite{PYJLC11} contains an
analysis of a $6$-state variant of Wiesner's money scheme, obtaining a
tight bound of $(2/3)^n$ on optimal simple counterfeiting attacks.
In addition, they show that the scheme can be made
error-tolerant---the bank will accept a bank note as long as say
$99\%$ of the qubit measurements are correct, allowing for the money
state to be slightly perturbed and still undergo a successful
authentication.\footnote{Our analysis can also be extended to this setting;
see Section~\ref{sec:par-rep} for more details.}
They also consider a classical-verification variant of the scheme
that is similar to (but somewhat less efficient than) the one we
propose, obtaining exponential security guarantees.

Other works consider more general counterfeiting attacks than we do,
and develop techniques that may be useful to extend our own results.  
In particular, Aaronson and Christiano~\cite{AC12} reduce security
against general $m\mapsto m+1$ cloners (given $m$ copies of a bank note,
produce $m+1$ quantum states that will be simultaneously accepted by the bank's 
verification procedure) to security against simple counterfeiting attacks of
the type we consider (attackers on their ``mini-schemes''). 
Pastawski et al.~\cite{PYJLC11} show that 
auxiliary access to the bank's verification procedure does not help,
\emph{provided} the only information returned by the bank is a single
bit, indicating success or failure. 
Indeed, intuitively this situation may be reduced to one in which the
cloner has no access to such a verification oracle simply by guessing:
because most attempts in verification will result in failure (otherwise
we would already have a successful cloner), the bits returned do not
contain much information.

\paragraph{Organization of the paper.}
We start with some preliminaries on quantum information theory and
semidefinite programming in Section~\ref{sec:preliminaries}.
Section~\ref{sec:generalized-schemes} contains our results on
Wiesner's quantum money scheme and generalizations, while
Section~\ref{sec:classical-ver} describes our results on schemes with
classical verification procedure.

\section{Preliminaries}
\label{sec:preliminaries}

We assume the reader is familiar with the basics of quantum
information theory, and suggest Nielsen and Chuang~\cite{NielsenC00}
to those who are not.
The purpose of this section is to summarize some of the notation and
basic concepts we make use of, and to highlight a couple of concepts
that may be less familiar to some readers.
The lecture notes \cite{WatrousNotes} may be helpful to readers
interested in further details on these topics.

\subsection{Basic notation, states, measurements and channels}

For any finite-dimensional complex Hilbert space $\X$ we write
$\lin{\X}$ to denote the set of linear operators acting on $\X$, 
 $\herm{\X}$ to denote the set of Hermitian operators acting
on $\X$, $\pos{\X}$ to denote the set of positive
semidefinite operators acting on $\X$, $\pd{\X}$ to denote
the set of positive definite operators acting on $\X$, and 
$\density{\X}$ to denote the set of density operators acting on $\X$.
For Hermitian operators $A,B\in\herm{\X}$ the notations $A\geq B$ and
$B\leq A$ indicate that $A - B$ is positive semidefinite, and the
notations $A > B$ and $B < A$ indicate that $A - B$ is positive definite.

Given operators $A,B\in\lin{\X}$, one defines the inner product
between $A$ and $B$ as $\ip{A}{B} = \tr(A^{\ast}B)$.
For Hermitian operators $A,B\in\herm{\X}$ it holds that
$\ip{A}{B}$ is a real number and satisfies $\ip{A}{B} = \ip{B}{A}$.
For every choice of finite-dimensional complex Hilbert spaces $\X$ and
$\Y$, and for a given linear mapping of the form
$\Phi:\lin{\X}\rightarrow\lin{\Y}$, there is a unique mapping
$\Phi^{\ast}:\lin{\Y}\rightarrow\lin{\X}$ (known as the \emph{adjoint}
of $\Phi$) that satisfies
$\ip{Y}{\Phi(X)} = \ip{\Phi^{\ast}(Y)}{X}$ for all $X\in\lin{\X}$ and
$Y\in\lin{\Y}$.

A \emph{register} is a hypothetical device that stores quantum
information.
Associated with a register $\reg{X}$ is a finite-dimensional complex
Hilbert space $\X$, and each quantum state of $\reg{X}$ is described
by a density operator $\rho\in\density{\X}$.
\emph{Qubits} are registers for which $\dim(\X) = 2$.
A \emph{measurement} of $\reg{X}$ is described by a set of positive
semidefinite operators $\{P_a\,:\,a\in\Sigma\}\subset\pos{\X}$,
indexed by a finite non-empty set of measurement outcomes $\Sigma$,
and satisfying the constraint $\sum_{a\in\Sigma}P_a = \I_{\X}$ (the
identity operator on $\X$).
If such a measurement is performed on $\reg{X}$ while it is in the
state $\rho$, each outcome $a\in\Sigma$ is obtained with probability
$\ip{P_a}{\rho}$.
A \emph{quantum channel} is a completely positive and trace-preserving
linear mapping of the form \mbox{$\Phi:\lin{\X}\rightarrow\lin{\Y}$} that
describes a hypothetical physical process that transforms each state
$\rho$ of a register $\reg{X}$ into the state $\Phi(\rho)$ of another
register $\reg{Y}$.
The identity channel that does nothing to a register $\reg{X}$ is
denoted $\I_{\lin{\X}}$.

\subsection{Linear mappings on spaces of operators}

Suppose $\op{dim}(\X) = d$ and assume that a fixed orthonormal basis
$\{\ket{1},\ldots,\ket{d}\}$ of $\X$ has been selected.
With respect to this basis, one defines the Choi-Jamio{\l}kowski
operator $J(\Phi)\in\lin{\Y\otimes\X}$ of a linear mapping
$\Phi:\lin{\X}\rightarrow\lin{\Y}$ as
\[
J(\Phi) = \sum_{1\leq i,j \leq d}
\Phi(\ket{i}\bra{j}) \otimes \ket{i}\bra{j}.
\]
The mapping $J$ is a linear bijection from the space of mappings of
the form $\Phi:\lin{\X}\rightarrow\lin{\Y}$ to 
$\lin{\Y\otimes\X}$.
It is well-known that $\Phi$ is completely positive if and only if
$J(\Phi) \in \pos{\Y\otimes\X}$, and that $\Phi$ is trace-preserving
if and only if $\tr_{\Y}(J(\Phi)) = \I_{\X}$
\cite{Choi75,Jamiolkowski72}. It is also well-known, and easy to verify, that
\begin{equation}\label{eq:jphi}
\bra{\phi} \Phi(\ket{\psi}\bra{\psi}) \ket{\phi}
= \bra{\phi \otimes \overline{\psi}} J(\Phi)\ket{\phi \otimes
  \overline{\psi}}
\end{equation}
for any choice of vectors $\ket{\psi}\in\X$ and
$\ket{\phi}\in\Y$, with
complex conjugation taken with respect to the standard basis.

\subsection{Semidefinite programming}

Semidefinite programming is a topic that has found several interesting
applications within quantum computing and quantum information theory
in recent years.
Here, we provide just a brief summary of semidefinite programming that
is focused on the narrow aspects of it that we use.
More comprehensive discussions can be found in
\cite{VandenbergheB96,Lovasz03,deKlerk02,BoydV04}, for instance.

A semidefinite program is a triple $(\Phi,A,B)$, where
\begin{mylist}{\parindent}
\item[1.] 
$\Phi: \lin{\X} \rightarrow \lin{\Y}$ is a Hermiticity-preserving
  linear mapping, and
\item[2.] $A\in\herm{\X}$ and $B\in\herm{\Y}$ are Hermitian operators,
\end{mylist}
for some choice of finite-dimensional complex Hilbert spaces $\X$ and $\Y$.
We associate with the triple $(\Phi,A,B)$ two optimization problems,
called the \emph{primal} and \emph{dual} problems, as follows:
\begin{center}
  \begin{minipage}{2.6in}
    \centerline{\underline{Primal problem}}\vspace{-7mm}
    \begin{align*}
      \text{maximize:}\quad & \ip{A}{X}\\
      \text{subject to:}\quad & \Phi(X) = B,\\
      & X\in\pos{\X}.
    \end{align*}
  \end{minipage}
  \hspace*{13mm}
  \begin{minipage}{2.6in}
    \centerline{\underline{Dual problem}}\vspace{-7mm}
    \begin{align*}
      \text{minimize:}\quad & \ip{B}{Y}\\
      \text{subject to:}\quad & \Phi^{\ast}(Y) \geq A,\\
      & Y\in\herm{\Y}.
    \end{align*}
  \end{minipage}
\end{center}
\noindent
The optimal primal value of this semidefinite program is
\[
\alpha = \sup\{\ip{A}{X}\,:\,X\in\pos{\X},\,\Phi(X) = B\},
\]
and the optimal dual value is
\[
\beta = \inf\{\ip{B}{Y}\,:\,Y\in\herm{\Y},\,\Phi^{\ast}(Y) \geq A\}.
\]
(It is to be understood that the supremum over an empty set is
$-\infty$ and the infimum over an empty set is $\infty$, so $\alpha$
and $\beta$ are well-defined values in 
$\real\cup\{-\infty,\infty\}$.
In this paper, however, we will only consider semidefinite programs for
which $\alpha$ and $\beta$ are finite.)

It always holds that $\alpha \leq \beta$, which is a fact known as
\emph{weak duality}.
The condition $\alpha = \beta$, which is known as 
\emph{strong duality}, does not hold for every semidefinite program,
but there are simple conditions known under which it does hold.
The following theorem provides one such condition (that has both a
primal and dual form).

\begin{theorem}[Slater's theorem for semidefinite programs]
  \label{theorem:Slater}
Let $(\Phi,A,B)$ be a semidefinite program and let $\alpha$ and
$\beta$ be its optimal primal and dual values.
\begin{mylist}{\parindent}
\item[1.]
  If $\beta$ is finite and there exists a positive definite operator
  $X\in\pd{\X}$ for which $\Phi(X) = B$,
  then $\alpha = \beta$ and there exists an operator $Y\in\herm{\Y}$
  such that $\Phi^{\ast}(Y)\geq A$ and $\ip{B}{Y} = \beta$.
\item[2.]
  If $\alpha$ is finite and there exists a Hermitian operator
  $Y\in\herm{\Y}$ for which $\Phi^{\ast}(Y) > A$,
  then $\alpha = \beta$ and there exists a positive semidefinite 
  operator $X\in\pos{\X}$ such that $\Phi(X)=B$ and 
  $\ip{A}{X} = \alpha$.
\end{mylist}
\end{theorem}

In words, the first item of this theorem states that if the dual
problem is feasible and the primal problem is 
\emph{strictly feasible}, then strong duality holds and the optimal
dual solution is achievable.
The second item is similar, with the roles of the primal and dual
problems reversed.

\section{Wiesner's quantum money and simple generalizations
  \label{sec:generalized-schemes}} 

Wiesner's quantum money scheme, and straightforward generalizations of
it, may be modeled in the following way.
An ensemble of pure quantum states
$\E = \left\{ (p_k,\ket{\psi_k})\,:\,k = 1,\ldots,N\right\}$
is fixed, and assumed to be known to all (including any would-be
counterfeiters).
When preparing a bank note, the bank randomly selects a key
$k\in\{1,\ldots,N\}$ with probability $p_k$.
The bank note's quantum system is initialized to the state
$\ket{\psi_k}$, and the note is labeled by a unique serial number.
The bank records the serial number along with the secret key $k$.

When an individual wishes to verify a bank note, she brings it
to the bank.
The bank looks up the key $k$ and measures the note's quantum state
with respect to the projective measurement $\{\Pi,\I- \Pi\}$, for
$\Pi = \ket{\psi_k}\bra{\psi_k}$.
The measurement outcome associated with $\Pi$ causes the bank note to
be declared valid, while the outcome associated with $\I-\Pi$ causes
the bank note to be declared invalid.

A simple counterfeiting attack against a scheme of the form just
described attempts to create two copies of a bank note from one,
and is considered to be successful if both copies independently
pass the bank's verification procedure. 
We take the original bank note's quantum state to be stored in a
register $\reg{X}$ having associated Hilbert space $\X$.
The registers storing the quantum states corresponding to the two
copies of the bank note produced by a would-be counterfeiter will be
called $\reg{Y}$ and $\reg{Z}$.
The Hilbert spaces $\Y$ and $\Z$ associated with these registers
are taken to be isomorphic to $\X$, but will retain distinct
names for the sake of our analysis.

Mathematically speaking, a simple counterfeiting attack is described
by a quantum channel $\Phi$ transforming $\reg{X}$ to
$(\reg{Y},\reg{Z})$, taking the state $\rho\in\density{\X}$ to the
state $\Phi(\rho)\in\density{\Y\otimes\Z}$.
In order to be physically realizable, at least in an idealized sense,
the channel $\Phi$ must correspond to a completely positive and trace
preserving linear mapping of the form 
$\Phi: \lin{\X}\rightarrow\lin{\Y\otimes\Z}$.
Conditioned on the bank having chosen the key $k$, the probability of
success for an attack described by $\Phi$ is given by
$\bra{\psi_k \otimes \psi_k} \Phi(\ket{\psi_k}\bra{\psi_k})
\ket{\psi_k\otimes\psi_k}$.
Averaging over the possible choices of $k$, 
 the overall success probability of a counterfeiting
attack is
\begin{equation} 
  \label{eq:probability-to-counterfeit}
  \sum_{k = 1}^N p_k
  \bra{\psi_k \otimes \psi_k} \Phi(\ket{\psi_k}\bra{\psi_k})
  \ket{\psi_k\otimes\psi_k}.
\end{equation}

\subsection{An SDP formulation of simple counterfeiting
  attacks \label{sec:counterfeit-sdp}} 

We now describe how the optimal success probability of a
counterfeiting strategy, which is represented by the supremum of the
probability \eqref{eq:probability-to-counterfeit} over all valid
channels $\Phi:\lin{\X}\rightarrow\lin{\Y\otimes\Z}$, may be
represented by a semidefinite program.
A similar semidefinite programming formulation may be found
in \cite{AudenaertD02,CerfF06,MolinaW11}, for instance.

The formulation makes use of the Choi-Jamio{\l}kowski representation 
$J(\Phi)$ of a given channel $\Phi$, as described in
Section~\ref{sec:preliminaries}.
Combining the characterization of all such representations that
correspond to quantum channels given there together
with~\eqref{eq:jphi} and the expression
\eqref{eq:probability-to-counterfeit}, it is not hard to see that the
optimal success probability of any simple counterfeiting strategy is
given by the following semidefinite program:\vspace{2mm}
\begin{center}
\begin{minipage}{0.5\textwidth}
  \centerline{\underline{Primal problem}}\vspace{-6mm}
  \begin{align*}
    \text{maximize:}\;\; & \ip{Q}{X}\\
    \text{subject to:}\;\; & \tr_{\Y\otimes\Z}(X) = \I_{\X}\\
    & X\in\pos{\Y\otimes\Z\otimes\X}
  \end{align*}
\end{minipage}\hspace*{5mm}
\begin{minipage}{0.4\textwidth}
  \centerline{\underline{Dual problem}}\vspace{-6mm}
  \begin{align*}
    \text{minimize:}\;\; & \tr(Y)\\
    \text{subject to:}\;\; & \I_{\Y\otimes\Z}\otimes Y \geq Q\\
    & Y \in \herm{\X}
  \end{align*}
\end{minipage}
\end{center}
where
\[
Q = \sum_{k = 1}^N p_k
\ket{\psi_k \otimes \psi_k \otimes \overline{\psi_k}}
\bra{\psi_k \otimes \psi_k \otimes \overline{\psi_k}}.
\]
(The dual problem is obtained from the primal problem in a routine way,
as described in Section~\ref{sec:preliminaries}.)

Because the primal and dual problems are both strictly feasible (as
follows by taking $X$ and $Y$ to be appropriately chosen multiples of
the identity, for example), it follows from
Theorem~\ref{theorem:Slater} that the optimal values for the primal
and dual problems are always equal, and are both achieved by feasible
choices for $X$ and $Y$.

\subsection{Analysis of Wiesner's original scheme (single-qubit case)}

To analyze Wiesner's original quantum money scheme, we begin by
considering the single-qubit (or $n=1$) case.
The analysis of the scheme for arbitrary values of $n$ will follow
from known results concerning \emph{product properties} of
semidefinite programs, as is described later in
Section~\ref{sec:par-rep}.

In the single-qubit case, Wiesner's quantum money scheme corresponds
to the ensemble
\[
\E = 
\left\{
\left(\frac{1}{4}, \, \ket{0}\right),
\left(\frac{1}{4}, \, \ket{1}\right),
\left(\frac{1}{4}, \, \ket{+}\right),
\left(\frac{1}{4}, \, \ket{-}\right)
\right\},
\]
which yields the operator 
\[
Q =  
\frac{1}{4} \left( \ket{000} \bra{000} + \ket{111}\bra{111} +  
\ket{+++} \bra{+++}  + \ket{---} \bra{---}   \right)
\]
in the semidefinite programming formulation described above.
We claim that the optimal value of the semidefinite program in this
case is equal to 3/4.
To prove this claim, it is sufficient to exhibit explicit primal and
dual feasible solutions achieving the value $3/4$.
For the primal problem, the value 3/4 is obtained by the solution
$X = J(\Phi)$, for $\Phi$ being the channel
\[
\Phi(\rho) = A_0 \rho A_0^{\ast} + A_1 \rho A_1^{\ast},
\]
where
\[
A_0 = \frac{1}{\sqrt{12}}
\begin{pmatrix} 3 & 0 \\ 0 & 1 \\  0 & 1 \\ 1 & 0 \end{pmatrix}
\qquad\text{and}\qquad
A_1 =  \frac{1}{\sqrt{12}}
\begin{pmatrix} 0 & 1 \\ 1 & 0 \\  1 & 0 \\ 0 & 3 \end{pmatrix}.
\]
For the dual problem, the value 3/4 is obtained by the solution
$Y = \frac{3}{8}\I_{\X}$, whose feasibility may be verified by computing
$\norm{Q} = 3/8$.

\subsection{Optimal single-qubit schemes}

It is natural to ask if the security of Wiesner's original scheme can
be improved through the selection of a different ensemble $\E$ in
place of the one considered in the previous section.
The answer is ``yes,'' as follows from our analysis of Wiesner's
original scheme together with the results of \cite{PYJLC11}, wherein
the authors consider the ensemble 
\[
\textstyle
\E = 
\left\{
\left(\frac{1}{6}, \, \ket{0}\right),
\left(\frac{1}{6}, \, \ket{1}\right),
\left(\frac{1}{6}, \, \ket{+}\right),
\left(\frac{1}{6}, \, \ket{-}\right),
\left(\frac{1}{6}, \, \frac{\ket{0} + i\ket{1}}{\sqrt{2}}\right),
\left(\frac{1}{6}, \, \frac{\ket{0} - i\ket{1}}{\sqrt{2}}\right)
\right\}.
\]
The operator $Q$ that one obtains is given by
\begin{equation} \label{eq:Q-symmetric}
  Q = \frac{1}{\op{rank}(\Pi)}
  \left(\I_{\lin{\Y}}\otimes\I_{\lin{\Z}}\otimes\op{T}\right)(\Pi)
\end{equation}
for $\Pi$ being the projection onto the symmetric subspace of
$\Y\otimes\Z\otimes\X$ and $\op{T}$ being the transposition
mapping with respect to the standard basis of $\X$.

The optimal value of the corresponding semidefinite program is 2/3.
Indeed, a primal feasible solution achieving the value $2/3$ is given by
$X = J(\Phi)$ for $\Phi$ being the channel
\[
\Phi(\rho) =  A_0 \rho A_0^{\ast} + A_1 \rho A_1^{\ast},
\]
where
\[
A_0 = \frac{1}{\sqrt{6}}
\begin{pmatrix}
  2 & 0 \\ 
  0 & 1 \\  
  0 & 1 \\ 
  0 & 0 
\end{pmatrix}
\qquad\text{and}\qquad
A_1 = \frac{1}{\sqrt{6}} 
\begin{pmatrix}
  0 & 0 \\ 
  1 & 0 \\  
  1 & 0 \\ 
  0 & 2 
\end{pmatrix}.
\]
(This channel is the optimal qubit cloner of Bu\v{z}ek and Hillery
\cite{BuzekH96}. )
A dual feasible solution achieving the bound 2/3 is given by
$Y = \frac{1}{3}\I_{\X}$ (with this solution's feasibility following from a
calculation of $\norm{Q} = 1/3$).

It is interesting to note that the same bound 2/3 can be obtained by a
four-state ensemble
\[
\textstyle
\E = \left\{
\left(\frac{1}{4}, \, \ket{\tau_1}\right),
\left(\frac{1}{4}, \, \ket{\tau_2}\right),
\left(\frac{1}{4}, \, \ket{\tau_3}\right),
\left(\frac{1}{4}, \, \ket{\tau_4}\right)
\right\},
\]
where $\{\ket{\tau_1},\ldots\ket{\tau_4}\}$ are any four states
forming a single qubit SIC-POVM
\cite{RenesBSC04}.
The operator $Q$ corresponding to any such ensemble is identical to the
one \eqref{eq:Q-symmetric} from the six-state ensemble above, and
therefore yields the same optimal value for the semidefinite program.

The schemes just mentioned are the best possible single qubit schemes.
To see this, one may simply consider the performance of 
$\Phi$ (i.,e., the Bu\v{z}ek--Hillery cloner), for which it
follows by a direct calculation that
\[
\bra{\psi\otimes\psi} \Phi(\ket{\psi}\bra{\psi}) \ket{\psi\otimes\psi}
= \frac{2}{3}
\]
for every state $\ket{\psi}$.
This shows that the optimal primal value, and therefore the optimal
counterfeiting probability, is always at least~2/3.

\subsection{Parallel repetitions of generalized Wiesner
  schemes \label{sec:par-rep}}

Wiesner's original scheme may be viewed as the $n$-fold parallel
repetition of a scheme wherein the spaces $\X$, $\Y$, and $\Z$
each represent a single qubit, and where the initial state of each
bank note is a state chosen uniformly from the set
$\{\ket{0},\ket{1},\ket{+},\ket{-}\}$.
That is, the preparation and verification of each $n$-qubit bank note
is, from the bank's perspective, equivalent to the \emph{independent}
preparation and verification of $n$ single-qubit bank notes; and a
successful counterfeiting attack is equivalent to a successful
counterfeiting attack against all $n$ of the single-qubit notes.
The value of $n$ plays the role of a security parameter, given that it becomes
increasingly hard to successfully counterfeit $n$ single-qubit bank
notes in a row, without failure, as $n$ grows large.

Now, there is nothing that forces a counterfeiter to attempt to
counterfeit an $n$-qubit bank note by treating each of its $n$ qubits
independently.
However, it is easily concluded from the semidefinite programming
formulation above that a counterfeiter gains no advantage whatsoever
by correlating multiple qubits during an attack.
This, in fact, is true for arbitrary choices of the ensemble $\E$, as
follows from a general result of Mittal and Szegedy
\cite{MittalS07} regarding \emph{product properties} of some
semidefinite programs.
(In our case, this property follows from the fact that the operator
$Q$ defining the objective function in the primal problem is always
positive semidefinite.)

In greater detail, let us consider the $n$-fold repetition of a
scheme, in which a single repetition of the scheme gives rise to a
semidefinite program determined by $Q\in\pos{\Y\otimes\Z\otimes\X}$.
Let us write $\X_j$, $\Y_j$, and $\Z_j$ to denote copies of the spaces
$\X$, $\Y$, and $\Z$ that represent the $j$-th repetition of the
scheme, for $j = 1,\ldots,n$, and let us write
$\X^{\otimes n} = \X_1\otimes\cdots\otimes\X_n$,
$\Y^{\otimes n} = \Y_1\otimes\cdots\otimes\Y_n$, and
$\Z^{\otimes n} = \Z_1\otimes\cdots\otimes\Z_n$.
The semidefinite program that describes the optimal simple
counterfeiting attack probability for the $n$-fold repetition is as
follows:
\begin{center}
\begin{minipage}{0.45\textwidth}
  \centerline{\underline{Primal problem}}\vspace{-6mm}
  \begin{align*}
    \text{maximize:}\;\; & \ip{W_\pi(Q^{\otimes n}) W_\pi^{\ast}}{X}\\
    \text{subject to:}\;\; & \tr_{\Y^{\otimes n}\otimes\Z^{\otimes
        n}}(X) = \I_{\X^{\otimes n}}\\
    & X\in\pos{\Y^{\otimes n}\otimes\Z^{\otimes n}\otimes\X^{\otimes n}}
  \end{align*}
\end{minipage}\hspace*{4mm}
\begin{minipage}{0.45\textwidth}
  \centerline{\underline{Dual problem}}\vspace{-6mm}
  \begin{align*}
    \text{minimize:}\;\; & \tr(Y)\\
    \text{subject to:}\;\; & \I_{\Y^{\otimes n}\otimes\Z^{\otimes
        n}}\otimes Y \geq
    W_\pi(Q^{\otimes n}) W_\pi^{\ast}\\
    & Y \in \herm{\X^{\otimes n}}
  \end{align*}
\end{minipage}
\end{center}

\noindent
In this semidefinite program, $\W_{\pi}$ is a unitary operator
representing a permutation of Hilbert spaces:
\begin{multline*}
  \qquad
  W_{\pi}
  \ket{(y_1\otimes z_1\otimes x_1)\otimes \cdots \otimes 
    (y_n\otimes z_n\otimes x_n)} \\
  = \ket{(y_1\otimes\cdots\otimes y_n)
    \otimes
    (z_1\otimes\cdots\otimes z_n)
    \otimes
    (x_1\otimes\cdots\otimes x_n)},
  \qquad
\end{multline*}
for all choices of $\ket{x_j}\in\X_j$, $\ket{y_j}\in\Y_j$, and
$\ket{z_j}\in\Z_j$, for $j = 1,\ldots,n$.

If the optimal value of the semidefinite program is $\alpha$ in the
single-repetition case, then the optimal value of the semidefinite
program for the $n$-fold repetition case is necessarily $\alpha^n$.
This may be proved by considering the primal and dual solutions
$X = W_{\pi}(X_1\otimes\cdots\otimes X_n)W_{\pi}^{\ast}$
and $Y = Y_1\otimes\cdots\otimes Y_n$, for $X_1,\ldots,X_n$ being
optimal primal solutions and $Y_1,\ldots,Y_n$ being optimal dual
solutions for the single-repetition semidefinite program.
The values obtained by these solutions are both $\alpha^n$.
Primal feasibility of $X$ is straightforward, while dual feasibility
of $Y$ follows from the fact that $A \geq B \geq 0$ implies
$A^{\otimes n}\geq B^{\otimes n}$ for all positive semidefinite $A$
and $B$.

\subsection{Threshold results}
 \label{sec:threshold}

One may also consider noise-tolerant variants of Wiesner's 
scheme, as was done in \cite{PYJLC11}.
In the setting discussed in the previous subsection
where $n$ repetitions of a particular scheme are performed, we may
suppose that the bank's verification procedure declares a bank note
valid whenever at least $t$ out of $n$ repetitions succeed, for some
choice of $t<n$, as opposed to requiring that all $n$ repetitions
succeed.  

One might hope that a similar analysis to the one in the previous
subsection will lead to an optimal counterfeiting probability of
\begin{equation} \label{eq:threshold-value}
\sum_{t\leq j\leq n}
\binom{n}{j}\alpha^j(1-\alpha)^{n-j}
\end{equation}
for such a scheme, for $\alpha$ being the optimal counterfeiting
probability for a single repetition.
This is the probability of successful counterfeiting when each
repetition is attacked independently.
In general, however, this bound may not be correct: the main result of
\cite{MolinaW11} demonstrates a related setting in which an analogous
bound does not hold, and explains the obstacle to obtaining such a
bound in general.
However, for some schemes, including Wiesner's original scheme and all
of the other specific schemes (including the classical verification ones in 
Section~\ref{sec:simple-classical-ver}) discussed in this paper, this bound will
be correct.
Letting $d = \dim{\X}$, the specific assumptions that we require to
obtain the bound \eqref{eq:threshold-value} are that
\begin{equation} \label{eq:ensemble-state}
  \sum_{k=1}^N p_k \ket{\psi_k}\bra{\psi_k} = \frac{1}{d}\I,
\end{equation} 
and that $Y = \frac{\alpha}{d} \I_{\X}$ is an optimal dual solution to
the single-repetition semidefinite program (from which it follows
$\norm{Q}= \frac{\alpha}{d}$).

To prove that these requirements are sufficient, let us introduce the
following notation.
We will write $Q_1$ in place of $Q$ to denote the operator that
specifies the semidefinite program representing a successful
counterfeiting attack, and we will also define
\[
Q_0 = \sum_{k = 1}^N p_k \left( \I_{\Y\otimes\Z} -
\ket{\psi_k\otimes\psi_k}\bra{\psi_k\otimes\psi_k}\right)
\otimes \ket{\overline{\psi_k}}\bra{\overline{\psi_k}},
\]
which has a complementary relationship to $Q_1$;
it represents a failure to counterfeit in a given repetition.
The semidefinite program describing the optimal counterfeiting
probability for the $n$-fold repetition scheme, where successes in $t$
repetitions are required for a validation, is then as follows:
\begin{center}
\begin{minipage}{0.45\textwidth}
  \centerline{\underline{Primal problem}}\vspace{-6mm}
  \begin{align*}
    \text{maximize:}\;\; & \ip{W_\pi R W_\pi^{\ast}}{X}\\
    \text{subject to:}\;\; & \tr_{\Y^{\otimes n}\otimes\Z^{\otimes
        n}}(X) = \I_{\X^{\otimes n}}\\
    & X\in\pos{\Y^{\otimes n}\otimes\Z^{\otimes n}\otimes\X^{\otimes n}}
  \end{align*}
\end{minipage}\hspace*{4mm}
\begin{minipage}{0.45\textwidth}
  \centerline{\underline{Dual problem}}\vspace{-6mm}
  \begin{align*}
    \text{minimize:}\;\; & \tr(Y)\\
    \text{subject to:}\;\; & \I_{\Y^{\otimes n}\otimes\Z^{\otimes
        n}}\otimes Y \geq W_\pi R W_\pi^{\ast}\\
    & Y \in \herm{\X^{\otimes n}}
  \end{align*}
\end{minipage}
\end{center}
where
\[
R = \sum_{\substack{a_1,\ldots,a_n\in\{0,1\}\\a_1 + \cdots + a_n\geq t}}
Q_{a_1} \otimes \cdots \otimes Q_{a_n}.
\]

To prove that the optimal value of this semidefinite program is given
by the expression \eqref{eq:threshold-value}, it suffices to exhibit
primal and dual feasible solutions achieving this value.
As for the standard $n$-fold repetition case described in the previous
subsection, it holds that $X =
W_{\pi}(X_1\otimes\cdots X_n)W_{\pi}^{\ast}$ 
is a primal feasible solution that achieves the desired value, where
again $X_1,\ldots,X_n$ are optimal primal solutions to the
single-repetition semidefinite program.
(This solution simply corresponds to an attacker operating
independently and optimally in each repetition.)
For the dual problem, we take
\[
Y = \norm{R} \I_{\X^{\otimes n}},
\]
which is clearly dual-feasible.
The condition \eqref{eq:ensemble-state} implies that
$Q_0 = \frac{1}{d}\I_{\Y\otimes\Z\otimes\X} - Q_1$, and a
consideration of spectral decompositions of the commuting operators
$Q_0$ and $Q_1$ reveals that
\[
\norm{R} = \frac{1}{d^n}
\sum_{t\leq j\leq n}
\binom{n}{j}\alpha^j(1-\alpha)^{n-j},
\]
which establishes the required bound.

\subsection{Optimal schemes in higher dimensions
  \label{sec:opt-higher-dimensions}}

We have observed that the best single-qubit variant of Wiesner's quantum
money scheme has an optimal counterfeiting probability of $2/3$, and
we know that the $n$-fold parallel repetition of this scheme has an optimal
counterfeiting probability of $(2/3)^n$.
Thus, bank notes storing a quantum state of dimension $d = 2^n$ can
have an optimal counterfeiting probability of $(2/3)^n$.
It is natural to ask whether one can do better, using a scheme that is
not given by the $n$-fold parallel repetition of a single qubit
scheme.

The answer is that there are better schemes (provided $n>1$).
More generally, for every $d$ representing the dimension of the state
stored by a quantum bank note, there exist schemes whose optimal
counterfeiting probability is equal to $2/(d+1)$, which is the best
that is possible: Werner's quantum cloning map \cite{Werner98} will
always succeed in counterfeiting any quantum bank note of dimension $d$
with probability $2/(d+1)$. 
The following proposition shows that there exists a scheme that matches
this bound in all dimensions~$d$. 

\begin{prop}\label{lem:opt-higher-dimensions}
Let $\mathcal{E}=\{p_k,\ket{\psi_k}\}$ be any ensemble of
$d$-dimensional states for which the operator
\[
Q = \sum_{k = 1}^N p_k
\ket{\psi_k \otimes \psi_k \otimes \overline{\psi_k}}
\bra{\psi_k \otimes \psi_k \otimes \overline{\psi_k}}
\]
is given by
\begin{equation}
Q \,=\, 
\frac{1}{\op{rank}(\Pi)}
\left(\I_{\lin{\complex^d}}\otimes\I_{\lin{\complex^d}}\otimes
\op{T}\right)(\Pi),\label{eq:qsym2}
\end{equation}
where $\op{T}$ is the transposition mapping with respect to the
standard basis of $\complex^d$ and $\Pi$ is the orthogonal projector
on the symmetric subspace of 
$\complex^d \otimes \complex^d \otimes \complex^d$. 
Then no simple counterfeiting strategy can succeed
against the money scheme derived from $\mathcal{E}$ with probability
more than $2/(d+1)$.
\end{prop}

Before proving the proposition, we note that any ensemble $\mathcal{E}$
obtained from a complex projective $(3,3)$-design (also known as a
quantum $3$-design~\cite{AmbE07}) satisfies~\eqref{eq:qsym2}, and thus
leads to an optimal $d$-dimensional money scheme. This also suggests
that one might obtain more efficient schemes (i.e., involving less
possible states for each part of the note) with security properties
similar to the ones described here if approximate designs are
considered instead.

\begin{proof}[Proof of Proposition~\ref{lem:opt-higher-dimensions}.]
Because we are looking for an upper bound on the maximum
counterfeiting probability, it suffices to construct a good feasible
solution $Y$ to the dual SDP described in
Section~\ref{sec:counterfeit-sdp}. 
We will choose $Y = \|Q\| \I_{\X}$, which
is a feasible dual solution with corresponding objective value 
$\tr(Y) = d \|Q\|$. We indicate how results from~\cite{EggelingW01}
may be used to show that $\|Q\| = 2/(d(d+1))$, proving the
proposition.

The operator $Q$ commutes with all operators of the
form $U\otimes U \otimes \overline{U}$, where $U$ is any unitary
acting on $\complex^d$. In Section VI.A of~\cite{EggelingW01} it is shown that
any such operator can be written as a linear combination of six 
conveniently chosen
Hermitian operators $S_+,S_-,S_0,S_1,S_2,S_3$ (for a definition see
Eqs.~(25a)--(25f) of~\cite{EggelingW01}). 
For our operator $Q$ we obtain the decomposition
\begin{equation}\label{eq:qs}
Q\,=\, \frac{1}{\op{rank}(\Pi)}\Big( \frac{1}{3}S_+ + \frac{d+2}{6}
\big(S_0 + S_1\big)\Big),
\end{equation}
where 
\begin{gather*}
S_+ = \frac{\I + V}{2} - \frac{1}{2(d+1)}\big( X +XV+VX + VXV  \big),\\
S_0+S_1 = \frac{1}{d+1}\big(X+XV+VX+VXV\big),
\end{gather*}
$V$ is the operator that permutes the first two registers on which $Q$
acts, and $X$ the partial transpose of the operator permuting the last
two registers.  
Moreover, as shown in~\cite{EggelingW01}, $S_+$ and $S_0$ are mutually
orthogonal projections, $S_0 S_1 = S_1 S_0 = S_1$, 
$S_+ S_1 = S_1 S_+ = 0$, and $S_1^2 = S_0$.
Hence, the decomposition~\eqref{eq:qs} shows that the operator norm of
$Q$ satisfies  
\[
\|Q\|
\, = \,
\frac{1}{\op{rank}(\Pi)}\frac{d+2}{3}
\,=\, \frac{2}{d(d+1)},
\]
as $\op{rank}(\Pi) = \binom{d+2}{3}$. 
\end{proof}

\section{Money schemes with classical verification}
\label{sec:classical-ver}

In this section we introduce a natural variant of Wiesner's scheme, as well as
higher-dimensional generalizations of it, in which the verification is
done through \emph{classical} communication with the bank. 
To distinguish the corresponding bank notes from the ones discussed in
the previous section, we will call them \emph{tickets}.\footnote{%
  As we will see, successful verification of a ticket necessarily
  entails its destruction. 
  This is unavoidable, as shown in~\cite{Gav11}. 
  To avoid this issue one may concatenate many tickets together to
  create a single bill, that will be able to go through as many 
	verification attempts as it contains tickets.}  

\subsection{Description of quantum tickets}\label{sec:class-tickets}

A quantum ticket is defined in the same way as a bank note: it is a
quantum state $\ket{\psi_k}$, where $k$ is a secret key kept by the
bank, together with a unique serial number. 
We consider schemes
in which the classical verification procedure has the following simple
form. The user first identifies herself to the bank by announcing her
ticket's serial number. The bank then sends her a classical
``challenge'' $c\in C$ chosen uniformly at random, where $C$ is some
fixed finite set. Depending on $c$, an honest user will perform a
measurement $\Pi_c = \{\Pi_c^a\}_{a\in A}$ on her ticket, and report
the outcome $a$ to the bank. The bank then looks up the secret key $k$
associated with the user's ticket, and accepts $a$ if and only if the
triple $(a,c,k)$ falls in a fixed, publicly known set $S$ of valid
triples.\footnote{For instance, the bank could accept all
  ``plausible'' answers, i.e., all $a$ such that
  $\bra{\psi_k}\Pi_c^a\ket{\psi_k} >0$. This condition ensures that
  honest users are always accepted.} 

A simple counterfeiting attack against such a scheme will attempt to
use just \emph{one} quantum ticket in order to successfully answer
\emph{two} independent challenges from the bank. Such a counterfeiter
may be modeled by a collection of POVMs $A_{c_1c_2} =
\{A_{c_1c_2}^{a_1a_2}\}_{a_1a_2}$, and its success probability is
\begin{equation}\label{eq:class-count}
  \sum_{k=1}^N p_k\, \frac{1}{|C|^2}\sum_{c_1,c_2} \,
  \sum_{\substack{(a_1,a_2):\\(a_1,c_1,k)\in S\\(a_2,c_2,k)\in S}} 
  \bra{\psi_k} A_{c_1 c_2}^{a_1 a_2} \ket{\psi_k},
\end{equation}
which is the ``classical-verification'' analogue
of~\eqref{eq:probability-to-counterfeit}.
By letting registers $\reg{Y}$ and
$\reg{Z}$ contain the answers $a_1$ and $a_2$ respectively, and
$\reg{X}$ contain the counterfeiter's input (the state $\ket{\psi_k}$
and the two challenges $c_1,c_2$), the problem of
maximizing~\eqref{eq:class-count} over all possible counterfeiting
strategies can be cast as a semidefinite program of the same form as
the one introduced in Section~\ref{sec:counterfeit-sdp}, with the
corresponding operator $Q$ defined as
$$ Q = \sum_{k=1}^N p_k\, \frac{1}{|C|^2} \sum_{c_1,c_2}
\,\sum_{\substack{(a_1,a_2):\\(a_1,c_1,k)\in S\\(a_2,c_2,k)\in S}}
\ket{a_1}\ket{a_2}\ket{c_1,c_2,\psi_k}
\bra{a_1}\bra{a_2}\bra{c_1,c_2,\psi_k}.$$ 
Since $Q$ is diagonal on the first $4$ registers, without
loss of generality an optimal solution $X$ to the primal problem will
be correspondingly block-diagonal, 
$$X \,=\, \sum_{a_1,a_2,c_1,c_2}
\ket{a_1,a_2,c_1,c_2}\bra{a_1,a_2,c_1,c_2} \otimes
X_{c_1c_2}^{a_1a_2},$$
and the SDP constraints are immediately seen to exactly enforce that
$\{X_{c_1c_2}^{a_1a_2}\}_{a_1a_2}$ is a POVM for every $(c_1,c_2)$.

We note that the problem faced by the counterfeiter can be cast as
a special instance of the more general \emph{state discrimination problem}.
Indeed, the counterfeiter's goal is to distinguish between the following: for every
pair of possible answers $(a_1,a_2)$, there is a mixed state corresponding
to the mixture over all states $\ket{c_1}\ket{c_2}\ket{\Psi_k}$
that for which $(a_1,a_2)$ would be a valid answer.
(Each state is weighted proportionally to the probability of the pair $(c_1,c_2)$
of being chosen as challenges by the bank, and of $\ket{\Psi_k}$ being chosen
as a bank note.) 
 As such, the fact that the optimal
counterfeiting strategy can be cast as a semidefinite program follows from
similar formulations for the general state discrimination problem (as the ones considered
in e.g.~\cite{EldarMV03}).

\subsection{Analysis of a simple class of qudit schemes}
\label{sec:simple-classical-ver}

We further restrict our attention to a natural class of extensions of
the classical-verification variant of Wiesner's scheme described in
the introduction. The schemes we consider are parametrized by a
dimension $d$ and two fixed bases
$\big\{\ket{e_0^0},\ldots,\ket{e_{d-1}^0}\big\}$ and
$\big\{\ket{e_0^1},\ldots,\ket{e_{d-1}^1}\big\}$ of $\complex^d$.\footnote{
It is easy to see that increasing the number of bases will only result
in weaker security: indeed, the more the bases the less likely it is
that the bank's randomly chosen challenge will match the basis used to 
encode each qudit.}
Each scheme is defined as the $n$-fold parallel repetition of a basic
scheme in which $N = 2d$, the states $|\psi_{(t,b)}\rangle$ are the
$\ket{e_t^b}$ for $t\in \{0,\ldots,d-1\}$ and $b\in\{0,1\}$, the random challenge
is a bit $c\in\{0,1\}$, and the valid answers are $a=t$ if $b=c$, and
any $a$ if $b\neq c$. Valid answers may be provided by an honest user
who measures his ticket in the basis corresponding to $c$. By writing
out the corresponding operator $Q$ and constructing a feasible
solution to the dual SDP, we show the following lemma, from which
Theorem~\ref{thm:main-class} follows directly. 

\begin{lemma}\label{lem:optdual-class} 
  For every simple counterfeiting attack against the $n$-qudit
  classical-verification scheme described above, the success
  probability is at most $\big(\frac{3}{4}+ \frac{\sqrt{c}}{4}\big)^n$,
  where $c = \max_{s,t} \left| \left\langle e_s^0|e_{t}^1
  \right\rangle\right|^2$
  is the \emph{effective overlap}.\footnote{%
    For any two bases of $\complex^d$,
    $c\geq 1/d$,
    and this is achieved for a pair of mutually unbiased bases. 
    This quantity also arises naturally in the study of uncertainty
    relations (see e.g.~\cite{TomR11}), of which our result may be
    seen as giving a special form.}
  \linebreak
  If $d=2$, there is always a counterfeiting strategy that
  achieves this bound.
\end{lemma}

\begin{proof}
We first analyze simple counterfeiting attacks against the basic
single-qudit scheme. 
Note that if both challenges from the bank are identical, the
counterfeiter can answer both correctly with probability $1$ by making
the appropriate measurement on his qubit. 

By symmetry, it suffices to consider the case where the first
challenge is $c_1 = 0$ and the second is $c_2 = 1$. 
In this case the operator $Q$ becomes
\begin{align*}
 Q &= \frac{1}{2d}\sum_{s,t=0}^{d-1}  
 \ket{s}\bra{s}_{\Y} \otimes \ket{t}\bra{t}_{\Z}
 \otimes
 \big(\ket{e_s^0}\bra{e_s^0}_\X + |e_t^1\rangle\langle e_t^1|_\X\big). 
\end{align*}
For $s,t\in\{0,\ldots,d-1\}$, let $V_{s,t} =
\ket{e_s^0}\bra{e_s^0}_\X  + \ket{e_{t}^1}\bra{e_{t}^1}_\X $. As $Q$ is block-diagonal, the dual SDP is
    \begin{align}
      \text{minimize:}\quad & \tr\big(Y\big)\notag\\
      \text{subject to:}\quad & Y \geq \frac{1}{2d} V_{s,t}
      \quad\text{(for all $s,t$)}
      \label{eq:sdp-class-1}\\
      & Y\in\herm{\complex^d}.\notag
    \end{align}
$V_{s,t}$ is a rank-$2$ Hermitian matrix whose eigenvalues are
    $1\pm\big|\langle e_s^0|e_{t}^1 \rangle\big|$. 
Hence, $Y = \frac{1+\sqrt{c}}{2d}\,\I$ is a feasible solution to the
dual problem with objective value $(1+\sqrt{c})/2$, leading to an
upper bound on the best counterfeiting strategy with overall success
probability at most $3/4+\sqrt{c}/4$. 

To finish the proof of the upper bound it suffices to note that the
SDP has the same parallel repetition property as was described in
Section~\ref{sec:par-rep}.

Finally, we show the ``moreover'' part of the claim.
Relabeling the vectors if necessary, assume
$|\langle e_0^0 | e_0^1\rangle|=\sqrt{c}$.
Let $\ket{u_0}$ be the eigenvector of $V_{0,0}$ with largest eigenvalue
$1+\sqrt{c}$, and $\ket{u_1}$ the eigenvector with smallest
eigenvalue. 
Using the observation that
$|\langle e_1^0 | e_1^1\rangle|=\sqrt{c}$, it
may be checked that
$$X =  \ket{0,0}\bra{0,0}\otimes \ket{u_0}\bra{u_0} +
\ket{1,1}\bra{1,1}\otimes \ket{u_1}\bra{u_1}$$
is a feasible solution to the primal SDP corresponding
to~\eqref{eq:sdp-class-1} (as expressed in
Section~\ref{sec:counterfeit-sdp}) with objective value
$(1+\sqrt{c})/2$, proving that the optimum of~\eqref{eq:sdp-class-1}
is exactly $(1+\sqrt{c})/2$.
\end{proof}

\subsection{A matching lower bound}

Let $d$ be a fixed dimension. 
We introduce a quantum ticket scheme for which the upper bound derived in the
previous section is tight. 
For $d=2$ our scheme recovers the one that is derived from Wiesner's
original quantum money.
Let $X_d$ and $Z_d$ be the generalized Pauli matrices, acting as 
$$X_d:\,\ket{i}\to\ket{i+1\mod d}\qquad\text{and}\qquad Z_d:\,\ket{i}\to \omega^{i}\ket{i},$$
 where 
$\omega = e^{2i\pi/d}$. 
Let $F$ be the quantum Fourier transform over $\integer_d$, 
$$F:\,\ket{i} \to \frac{1}{\sqrt{d}} \sum_j \omega^{ij} \ket{j},$$
and note that $X_d = FZ_d F^\dagger$. 
Let $\left\{\ket{e_t^0}\right\}$ be the basis defined by 
$\ket{e_t^0} = (X_d)^t\ket{0} = \ket{t}$, and $\{\ket{e_t^1}\}$
the Fourier-transformed basis $\ket{e_t^1} = F\ket{e_t^0} = (Z_d)^t
F\ket{0}$ for every $t$. 
Then 
$$\big|\langle e_s^0 | e_t^1\rangle\big| 
\,=\,\big|\langle s | F | t\rangle \big|\,=\,\frac{1}{\sqrt{d}}$$
 for every $s,t$: 
the corresponding overlap is $c=1/d$. 
Lemma~\ref{lem:optdual-class} shows that the optimal cloner achieves
success at most $3/4+1/(4\sqrt{d})$. The following lemma states a
matching lower bound.

\begin{lemma}\label{lem:class-lb} There is a cloner for the $n$-qudit
  ticket scheme described above which successfully answers both
  challenges with success probability
  $\big(\frac{3}{4}+\frac{1}{4\sqrt{d}}\big)^n$.
\end{lemma}

\begin{proof}
We describe a cloner that acts independently on each qudit, succeeding
with probability $\frac{3}{4}+\frac{1}{4\sqrt{d}}$ on each 
qudit.\footnote{%
  The analysis is very similar to one that was done in~\cite{VW11}, in
  a different context but for essentially the same problem.}
Let 
$$\ket{\psi} \,=\, \big(2+2/\sqrt{d}\big)^{-1/2}(\ket{0}+F\ket{0}),$$
and for every $(s,t)$ let $P_{s,t}$ be the rank 1 projector on the
unit vector $X_d^s Z_d^t\ket{\psi}$. 
As a consequence of Schur's lemma, 
$\sum_{s,t} \frac{1}{d}P_{s,t} = \I$, so that $\big\{P_{s,t}/d\big\}$
is a POVM.

The cloner proceeds as follows: if the challenge is either $00$ or
$11$, he measures in the corresponding basis and sends the resulting
outcome as answer to both challenges. In this case he is always
correct. In case the challenge is either $01$ or $10$, he measures the
ticket using the POVM $\{P_{s,t}/d\}$, and uses $s$ as answer to the
challenge ``$0$'' and $t$ as answer to the challenge ``$1$''. 
Because the two challenges are distinct, only one of them corresponds
to the actual basis in which the ticket was encoded. 
Without loss of generality assume this is the ``$0$'' basis, so that
the ticket is $\ket{e_{s}^0} = \ket{s}$.
The probability that the cloner obtains the correct outcome $s$ is
\begin{align*}
 \frac{1}{d}\sum_t \tr\big(P_{s,t} \ket{s}\bra{s} \big) 
 & = \frac{1}{d}\sum_t \big|\bra{s} X_d^s Z_d^t\ket{\psi}\big|^2\\
 & = \frac{1}{d}\sum_t \big|\bra{0} Z_d^t\ket{\psi}\big|^2\\
 & = \big|\langle 0 | \psi\rangle\big|^2,
\end{align*}
because, for every $t$, it holds that $\bra{0}Z_d^t =\omega^t\bra{0}$. 
To conclude, it suffices to compute
\[
\big|
\langle 0 | \psi \rangle \big|^2 
\,=\, \frac{1}{2+2/\sqrt{d}}\big|
\langle 0 | 0\rangle + \langle 0 | F | 0 \rangle
\big|^2\,=\,\frac{1}{2}\Big(1+\frac{1}{\sqrt{d}}\Big).
\]
\end{proof}

\paragraph{Acknowledgments}

We thank Scott Aaronson for his question\footnote{%
  \url{http://theoreticalphysics.stackexchange.com/questions/370/}} 
on Theoretical Physics Stack Exchange that originated the results in
this paper as an answer, and Peter Shor for pointing out
the connection between the channel representing an optimal
attack on Wiesner's quantum money, and the optimal cloners studied in
\cite{BruCDM00} and \cite{BuzekH96}. 
JW thanks Debbie Leung and Joseph Emerson for helpful discussions.
AM thanks Michael Grant and Stephen Boyd for their creation of the CVX software.
\bibliographystyle{alpha}
\bibliography{Wiesner}

\end{document}